\documentclass[12pt]{article}%
\usepackage{amsmath}
\usepackage{mitpress}
\usepackage{amsfonts}
\usepackage{amssymb}
\usepackage{graphicx}%
\providecommand{\U}[1]{\protect\rule{.1in}{.1in}}
\newtheorem{theorem}{Theorem}

\newtheorem{definition}[theorem]{Definition}
\newtheorem{example}[theorem]{Example}

\newtheorem{notation}[theorem]{Notation}

\newtheorem{remark}[theorem]{Remark}

\newenvironment{proof}[1][Proof]{\noindent\textbf{#1.} }{\ \rule{0.5em}{0.5em}}
\newdimen\dummy
\dummy=\oddsidemargin
\begin{document}

\title{Defining the Symmetry of the Universal Semi-regular Autonomous Asynchronous Systems}
\author{Serban E. Vlad\\Str. Zimbrului, Nr. 3, Bl. PB68, Ap. 11, 410430, Oradea, Romania, email: serban\_e\_vlad@yahoo.com}
\maketitle

\begin{abstract}
The regular autonomous asynchronous systems are the non-deterministic Boolean
dynamical systems and universality means the greatest in the sense of the
inclusion. The paper gives four definitions of symmetry of these systems in a
slightly more general framework, called semi-regularity and also many examples.

\end{abstract}

\textbf{MSC}: 94C10

\textbf{keywords}: asynchronous system, symmetry

\section{Introduction}

Switching theory has developed in the 50's and the 60's as a common effort of
the mathematicians and the engineers of studying the switching circuits
(=asynchronous circuits) from digital electrical engineering. After 1970 we do
not know to exist any mathematical published work in what we call switching
theory. The published works are written by engineers and their approach is
always descriptive and unacceptable for the mathematicians. The label of
\textit{switching theory} has changed to \textit{asynchronous systems} (or
\textit{circuits}) \textit{theory}. One of the possible motivations of the
situation consists in the fact that the important producers of digital
equipments have stopped the dissemination of such researches.

Our interest in asynchronous systems had bibliography coming from the 50's and
the 60's, as well as engineering works giving intuition, as well as
mathematical works giving analogies. An interesting \textit{rendezvous} has
happened when the asynchronous systems theory has met the dynamical systems
theory, resulting the so called \textit{regular} autonomous systems=Boolean
dynamical systems; the vector field is $\Phi:\{0,1\}^{n}\rightarrow
\{0,1\}^{n},$ time is discrete or real and we obtain the \textit{unbounded
delay model} of computation of $\Phi$, suggested by the engineers. The
\textit{synchronous} iterations of $\Phi:\Phi\circ\Phi,\Phi\circ\Phi\circ
\Phi,...$ of the dynamical systems are replaced by \textit{asynchronous}
iterations in which each coordinate $\Phi_{1},...,\Phi_{n}$ is iterated
independently on the others, in arbitrary finite time.

We denote with $\mathbf{B}=\{0,1\}$ the binary Boolean algebra, together with
the discrete topology and with the usual algebraical laws:%
\begin{align*}
&
\begin{array}
[c]{cc}
& \overline{\ \ }\\
0 & 1\\
1 & 0
\end{array}
,\;\;\;\;\;\;%
\begin{array}
[c]{ccc}%
\cdot & 0 & 1\\
0 & 0 & 0\\
1 & 0 & 1
\end{array}
,\;\;\;\;\;\;%
\begin{array}
[c]{ccc}%
\cup & 0 & 1\\
0 & 0 & 1\\
1 & 1 & 1
\end{array}
,\;\;\;\;\;\;%
\begin{array}
[c]{ccc}%
\oplus & 0 & 1\\
0 & 0 & 1\\
1 & 1 & 0
\end{array}
\\
&  \;\;\;\;\;\;\;\;\;\;\;\;\;\;\;\;\;\;\;\ \ \ \ \ Table\;1
\end{align*}
We use the same notations for the laws that are induced from $\mathbf{B}$ on
other sets, for example $\forall x\in\mathbf{B}^{n},\forall y\in\mathbf{B}%
^{n},$%
\[
\overline{x}=(\overline{x_{1}},...,\overline{x_{n}}),
\]%
\[
x\cup y=(x_{1}\cup y_{1},...,x_{n}\cup y_{n})
\]
etc. In Figure \ref{simetrie30}
\begin{figure}
[ptb]
\begin{center}
\fbox{\includegraphics[
height=1.6008in,
width=4.1027in
]%
{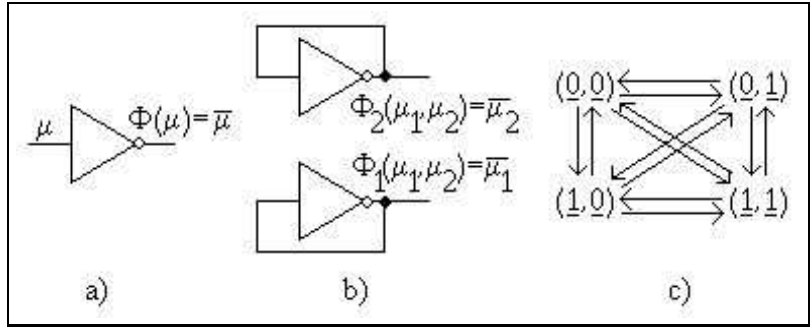}%
}\caption{a) the logical gate NOT, b) circuit with logical gates NOT, c) state
portrait}%
\label{simetrie30}%
\end{center}
\end{figure}
we have drawn at a) the logical gate NOT, i.e. the circuit that computes the
logical complement and at b) a circuit that makes use of logical gates NOT.
The asynchronous system that models the circuit from b) has the state portrait
drawn at c). In the state portraits, the arrows show the increase of (the
discrete or continuous) time. The underlined coordinates \underline{$\mu_{i}$}
are these coordinates for which $\Phi_{i}(\mu_{i})\neq\mu_{i}$ and they are
called \textit{excited}, or \textit{enabled}, or \textit{unstable}. The
coordinates $\mu_{i}$ that are not underlined fulfill by definition $\Phi
_{i}(\mu_{i})=\mu_{i}$ and they are called \textit{not excited}, or
\textit{not enabled}, or \textit{stable}. The existence of two underlined
coordinates in $(0,0)$ shows that $\Phi_{1}(0,0)=1$ may be computed first,
$\Phi_{2}(0,0)=1$ may be computed first, or $\Phi_{1}(0,0),\Phi_{2}(0,0)$ may
be computed simultaneously, thus when the system is in $(0,0),$ it may run in
three different directions, non-determinism.

Our present purpose is to define the symmetry of these systems.

\section{Semi-regular systems}

\begin{notation}
We denote $\mathbf{N}_{\_}=\{-1,0,1,2,...\}$.
\end{notation}

\begin{notation}
$\chi_{A}:\mathbf{R}\rightarrow\mathbf{B}$ is the notation of the
characteristic function of the set $A\subset\mathbf{R}$: $\forall
t\in\mathbf{R},\chi_{A}(t)=\left\{
\begin{array}
[c]{c}%
0,if\ t\notin A\\
1,if\ t\in A
\end{array}
\right.  .$
\end{notation}

\begin{notation}
We denote with $\overline{\Pi}_{n}$ the set of the sequences $\alpha
=\alpha^{0},\alpha^{1},...,$ $\alpha^{k},$ $...\in\mathbf{B}^{n}.$
\end{notation}

\begin{notation}
The set of the real sequences $t_{0}<t_{1}<...<t_{k}<...$ that are unbounded
from above is denoted with $Seq.$
\end{notation}

\begin{notation}
We use the notation $\overline{P}_{n}$ for the set of the functions
$\rho:\mathbf{R}\rightarrow\mathbf{B}^{n}$ having the property that $\alpha
\in\overline{\Pi}_{n}$ and $(t_{k})\in Seq$ exist with $\forall t\in
\mathbf{R},$%
\begin{equation}
\rho(t)=\alpha^{0}\chi_{\{t_{0}\}}(t)\oplus\alpha^{1}\chi_{\{t_{1}\}}%
(t)\oplus...\oplus\alpha^{k}\chi_{\{t_{k}\}}(t)\oplus... \label{sym16}%
\end{equation}

\end{notation}

\begin{definition}
\label{Def39}Let $\Phi:\mathbf{B}^{n}\rightarrow\mathbf{B}^{n}$ be a function.
For $\nu\in\mathbf{B}^{n},\nu=(\nu_{1},...,\nu_{n})$ we define the function
$\Phi^{\nu}:\mathbf{B}^{n}\rightarrow\mathbf{B}^{n}$ by $\forall\mu
\in\mathbf{B}^{n},$%
\[
\Phi^{\nu}(\mu)=(\overline{\nu_{1}}\mu_{1}\oplus\nu_{1}\Phi_{1}(\mu
),...,\overline{\nu_{n}}\mu_{n}\oplus\nu_{n}\Phi_{n}(\mu)).
\]

\end{definition}

\begin{definition}
\label{Def43}Let be $\alpha\in\overline{\Pi}_{n}.$ The function $\widehat
{\Phi}^{\alpha}:\mathbf{B}^{n}\times\mathbf{N}_{\_}\rightarrow\mathbf{B}^{n}$
defined by $\forall\mu\in\mathbf{B}^{n},\forall k\in\mathbf{N}_{\_},$%
\begin{equation}
\left\{
\begin{array}
[c]{c}%
\widehat{\Phi}^{\alpha}(\mu,-1)=\mu,\\
\widehat{\Phi}^{\alpha}(\mu,k+1)=\Phi^{\alpha^{k+1}}(\widehat{\Phi}^{\alpha
}(\mu,k))
\end{array}
\right.  \label{ite2}%
\end{equation}
is called \textbf{discrete time} $\alpha-$\textbf{semi-orbit} \textbf{of}
$\mu.$ We consider also the sequence $(t_{k})\in Seq$ and the function
$\rho\in\overline{P}_{n}$ from (\ref{sym16}), for which the function
$\Phi^{\rho}:\mathbf{B}^{n}\times\mathbf{R}\rightarrow\mathbf{B}^{n}$ is
defined by: $\forall\mu\in\mathbf{B}^{n},\forall t\in\mathbf{R},$%
\begin{equation}
\Phi^{\rho}(\mu,t)=\widehat{\Phi}^{\alpha}(\mu,-1)\chi_{(-\infty,t_{0}%
)}(t)\oplus\widehat{\Phi}^{\alpha}(\mu,0)\chi_{\lbrack t_{0},t_{1})}%
(t)\oplus\label{ite3}%
\end{equation}%
\[
\oplus\widehat{\Phi}^{\alpha}(\mu,1)\chi_{\lbrack t_{1},t_{2})}(t)\oplus
...\oplus\widehat{\Phi}^{\alpha}(\mu,k)\chi_{\lbrack t_{k},t_{k+1})}%
(t)\oplus...
\]
$\Phi^{\rho}$ is called \textbf{continuous time} $\rho-$\textbf{semi-orbit}
\textbf{of} $\mu.$
\end{definition}

\begin{definition}
\label{Def62_}The \textbf{discrete time} and the \textbf{continuous time
universal semi-regular autonomous asynchronous systems} are defined by%
\[
\widehat{\overline{\Xi}}_{\Phi}=\{\widehat{\Phi}^{\alpha}(\mu,\cdot)|\mu
\in\mathbf{B}^{n},\alpha\in\overline{\Pi}_{n}\},
\]%
\[
\overline{\Xi}_{\Phi}=\{\Phi^{\rho}(\mu,\cdot)|\mu\in\mathbf{B}^{n},\rho
\in\overline{P}_{n}\}.
\]

\end{definition}

\begin{remark}
$\widehat{\overline{\Xi}}_{\Phi},$ $\overline{\Xi}_{\Phi}$ and $\Phi$ are
usually identified.
\end{remark}

\begin{example}
\label{Exa2}In Figure \ref{simetrie23} we have drawn at a)
\begin{figure}
[ptb]
\begin{center}
\fbox{\includegraphics[
height=2.3488in,
width=3.8354in
]%
{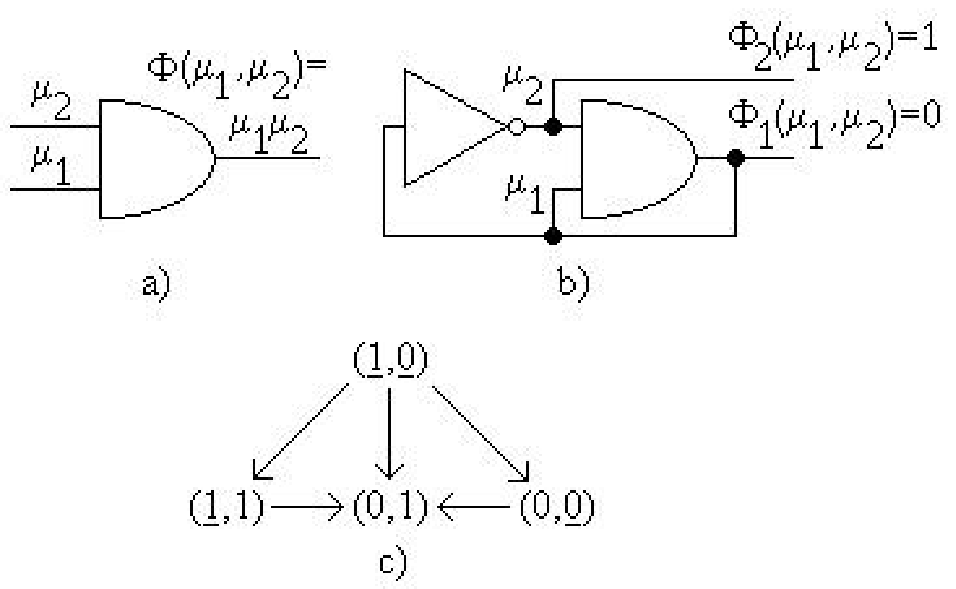}%
}\caption{The semi-regular system $\overline{\Xi}_{\Phi}$ from Example
\ref{Exa2}}%
\label{simetrie23}%
\end{center}
\end{figure}
the AND gate that computes the logical intersection, at b) a circuit with two
gates and at c) the state portrait of $\Phi:\mathbf{B}^{2}\rightarrow
\mathbf{B}^{2},\forall(\mu_{1},\mu_{2})\in\mathbf{B}^{2},\Phi(\mu_{1},\mu
_{2})=(0,1).$ We conclude that%
\[
\overline{\Xi}_{\Phi}=\{(\mu_{1},\mu_{2})\chi_{(-\infty,t_{0})}\oplus(\mu
_{1}\lambda_{1},\mu_{2}\cup\lambda_{2})\chi_{\lbrack t_{0},t_{1})}\oplus
\]%
\[
\oplus(\mu_{1}\lambda_{1}\nu_{1},\mu_{2}\cup\lambda_{2}\cup\nu_{2}%
)\chi_{\lbrack t_{1},\infty)}|\mu,\lambda,\nu\in\mathbf{B}^{2},t_{0},t_{1}%
\in\mathbf{R},t_{0}<t_{1}\},
\]
since the first coordinate might finally decrease its value and the second
coordinate might finally increase its value, but the order and the time
instant when these things happen are arbitrary.
\end{example}

\section{Anti-semi-regular systems}

\begin{definition}
\label{Def15}Let be $\Phi:\mathbf{B}^{n}\rightarrow\mathbf{B}^{n},\alpha
\in\overline{\Pi}_{n},(t_{k})\in Seq$ and $\rho\in\overline{P}_{n}$ from
(\ref{sym16}). The function $^{\ast}\widehat{\Phi}^{\alpha}:\mathbf{B}%
^{n}\times\mathbf{N}_{\_}\rightarrow\mathbf{B}^{n}$ is defined by: $\forall
\mu\in\mathbf{B}^{n},\forall k\in\mathbf{N}_{\_},$%
\begin{equation}
\left\{
\begin{array}
[c]{c}%
^{\ast}\widehat{\Phi}^{\alpha}(\mu,-1)=\mu,\\
\Phi^{\alpha^{k+1}}(^{\ast}\widehat{\Phi}^{\alpha}(\mu,k+1))=\text{ }^{\ast
}\widehat{\Phi}^{\alpha}(\mu,k)
\end{array}
\right.  \label{ite6}%
\end{equation}
and $^{\ast}\Phi^{\rho}:\mathbf{B}^{n}\times\mathbf{R}\rightarrow
\mathbf{B}^{n}$ is defined by: $\forall\mu\in\mathbf{B}^{n},\forall
t\in\mathbf{R},$%
\begin{equation}
^{\ast}\Phi^{\rho}(\mu,t)=\text{ }^{\ast}\widehat{\Phi}^{\alpha}(\mu
,-1)\chi_{(-\infty,t_{0})}(t)\oplus\text{ }^{\ast}\widehat{\Phi}^{\alpha}%
(\mu,0)\chi_{\lbrack t_{0},t_{1})}(t)\oplus\label{ite9}%
\end{equation}%
\[
\oplus\text{ }^{\ast}\widehat{\Phi}^{\alpha}(\mu,1)\chi_{\lbrack t_{1},t_{2}%
)}(t)\oplus...\oplus\text{ }^{\ast}\widehat{\Phi}^{\alpha}(\mu,k)\chi_{\lbrack
t_{k},t_{k+1})}(t)\oplus...
\]
$^{\ast}\widehat{\Phi}^{\alpha}$ is called \textbf{discrete time} $\alpha
-$\textbf{anti-semi-orbit} \textbf{of} $\mu,$ while $^{\ast}\Phi^{\rho}$ is
called \textbf{continuous time} $\rho-$\textbf{anti-semi-orbit} \textbf{of}
$\mu.$
\end{definition}

\begin{remark}
We compare the semi-orbits and the anti-semi-orbits now and see that they run
both from the past to the future, but the relation cause-effect is different:
in $\widehat{\Phi}^{\alpha},\Phi^{\rho}$ the cause is in the past and the
effect is in the future, while in $^{\ast}\widehat{\Phi}^{\alpha},$ $^{\ast
}\Phi^{\rho}$ the cause is in the future and the effect is in the past.
\end{remark}

\begin{definition}
\label{Def17}The \textbf{discrete time} and the \textbf{continuous time
universal anti-semi-regular autonomous asynchronous systems} are defined by%
\[
^{\ast}\widehat{\overline{\Xi}}_{\Phi}=\{^{\ast}\widehat{\Phi}^{\alpha}%
(\mu,\cdot)|\mu\in\mathbf{B}^{n},\alpha\in\overline{\Pi}_{n}\},
\]%
\[
^{\ast}\overline{\Xi}_{\Phi}=\{^{\ast}\Phi^{\rho}(\mu,\cdot)|\mu\in
\mathbf{B}^{n},\rho\in\overline{P}_{n}\}.
\]

\end{definition}

\begin{example}
\label{Exa16}In Figure \ref{simetrie24}
\begin{figure}
[ptb]
\begin{center}
\fbox{\includegraphics[
height=1.2436in,
width=4.1433in
]%
{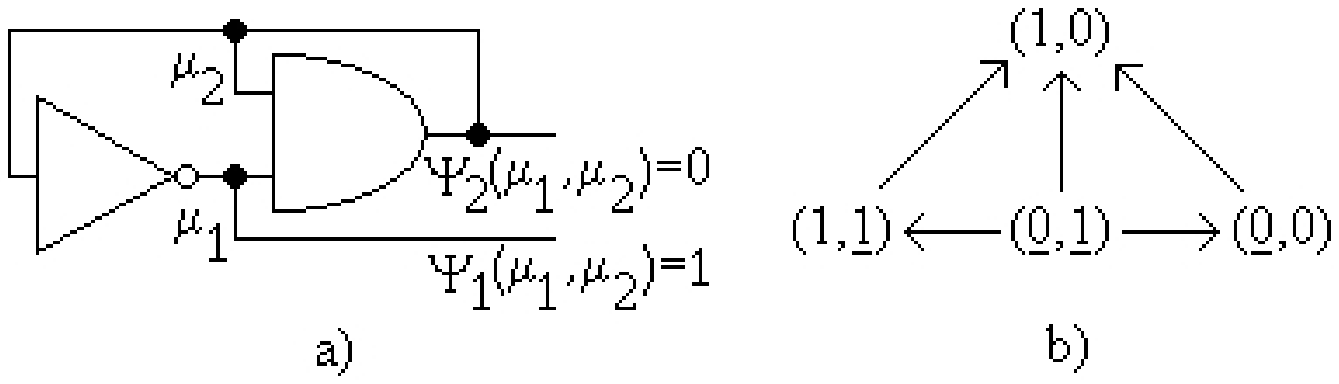}%
}\caption{The semi-regular system $\overline{\Xi}_{\Psi}^{\ast}$ from Example
\ref{Exa16}}%
\label{simetrie24}%
\end{center}
\end{figure}
we have drawn at a) the circuit and at b) the state portrait of $\Psi
:\mathbf{B}^{2}\rightarrow\mathbf{B}^{2},\forall(\mu_{1},\mu_{2})\in
\mathbf{B}^{2},$ $\Psi(\mu_{1},\mu_{2})=(1,0)$ for which%
\[
\overline{\Xi}_{\Psi}=\{(\mu_{1},\mu_{2})\chi_{(-\infty,t_{0})}\oplus(\mu
_{1}\cup\lambda_{1},\mu_{2}\lambda_{2})\chi_{\lbrack t_{0},t_{1})}\oplus
\]%
\[
\oplus(\mu_{1}\cup\lambda_{1}\cup\nu_{1},\mu_{2}\lambda_{2}\nu_{2}%
)\chi_{\lbrack t_{1},\infty)}|\mu,\lambda,\nu\in\mathbf{B}^{2},t_{0},t_{1}%
\in\mathbf{R},t_{0}<t_{1}\}.
\]
The arrows in Figures \ref{simetrie23} c) and \ref{simetrie24} b) are the
same, but with a different sense and we note that $\overline{\Xi}_{\Psi}=$
$^{\ast}\overline{\Xi}_{\Phi},$ where $\Phi$ is the one from Example
\ref{Exa2}.
\end{example}

\section{Isomorphisms and anti-isomorphisms}

\begin{definition}
\label{Def21}Let be $g:\mathbf{B}^{n}\rightarrow\mathbf{B}^{n}.$ It defines
the functions $\widehat{g}:\overline{\Pi}_{n}\rightarrow\overline{\Pi}%
_{n},\forall\alpha\in\overline{\Pi}_{n},\forall k\in\mathbf{N},$%
\[
\widehat{g}(\alpha)(k)=g(\alpha^{k});
\]
$\widetilde{g}:\overline{P}_{n}\rightarrow\overline{P}_{n},\forall\rho
\in\overline{P}_{n},\forall t\in\mathbf{R},$%
\[
\widetilde{g}(\rho)(t)=\left\{
\begin{array}
[c]{c}%
(0,...,0),if\;\rho(t)=(0,...,0),\\
g(\rho(t)),otherwise
\end{array}
\right.
\]
and $g:(\mathbf{B}^{n})^{\mathbf{R}}\rightarrow(\mathbf{B}^{n})^{\mathbf{R}%
},\forall x\in(\mathbf{B}^{n})^{\mathbf{R}},\forall t\in\mathbf{R},$%
\[
g(x)(t)=g(x(t)).
\]

\end{definition}

\begin{theorem}
\label{The29}Let be the functions $\Phi,\Psi,g,g^{\prime}:\mathbf{B}%
^{n}\rightarrow\mathbf{B}^{n}.$ The following statements are equivalent:

a) $\forall\nu\in\mathbf{B}^{n},$ the diagram%
\[%
\begin{array}
[c]{ccc}%
\mathbf{B}^{n} & \overset{\Phi^{\nu}}{\rightarrow} & \mathbf{B}^{n}\\
g\downarrow\; &  & \;\downarrow g\\
\mathbf{B}^{n} & \overset{\Psi^{g^{\prime}(\nu)}}{\rightarrow} &
\mathbf{B}^{n}%
\end{array}
\]
is commutative;

b) $\forall\mu\in\mathbf{B}^{n},\forall\alpha\in\overline{\Pi}_{n},\forall
k\in\mathbf{N}_{\_},$%
\[
g(\widehat{\Phi}^{\alpha}(\mu,k))=\widehat{\Psi}^{\widehat{g^{\prime}}%
(\alpha)}(g(\mu),k);
\]

c) $\forall\mu\in\mathbf{B}^{n},$%
\[
g(\mu)=\Psi^{g^{\prime}(0,...,0)}(g(\mu))
\]
and $\forall\mu\in\mathbf{B}^{n},\forall\rho\in\overline{P}_{n},\forall
t\in\mathbf{R},$%
\[
g(\Phi^{\rho}(\mu,t))=\Psi^{\widetilde{g^{\prime}}(\rho)}(g(\mu),t).
\]

\end{theorem}

\begin{proof}
a)$\Longrightarrow$b) We fix arbitrarily $\mu\in\mathbf{B}^{n},$ $\alpha
\in\overline{\Pi}_{n}$ and we use the induction on $k\geq-1$. For $k=-1$, b)
becomes $g(\mu)=g(\mu)$, thus we suppose that it is true for $k$ and we prove
it for $k+1$:%
\[
g(\widehat{\Phi}^{\alpha}(\mu,k+1))=g(\Phi^{\alpha^{k+1}}(\widehat{\Phi
}^{\alpha}(\mu,k)))=\Psi^{g^{\prime}(\alpha^{k+1})}(g(\widehat{\Phi}^{\alpha
}(\mu,k)))=
\]%
\[
=\Psi^{g^{\prime}(\alpha^{k+1})}(\widehat{\Psi}^{\widehat{g^{\prime}}(\alpha
)}(g(\mu),k))=\widehat{\Psi}^{\widehat{g^{\prime}}(\alpha)}(g(\mu),k+1).
\]

b)$\Longrightarrow$c) The first statement results from b) if we take
$\alpha^{0}=(0,...,0)$ and $k=0.$ In order to prove the second statement, let
$\mu\in\mathbf{B}^{n}$ and $\rho\in\overline{P}_{n}$ be arbitrary, thus
(\ref{sym16}) holds with $(t_{k})\in Seq,\rho(t_{0}),...,\rho(t_{k}%
),...\in\overline{\Pi}_{n}.$ If $\forall t\in\mathbf{R},\rho(t)=(0,...,0)$ the
statement to prove takes the form $g(\mu)=g(\mu)$ so that we can suppose now
that a finite or an infinite number of $\rho(t_{k})$ are $\neq(0,...,0).$ In
the case $\forall k\in\mathbf{N},\rho(t_{k})\neq(0,...,0)$ that does not
restrict the generality of the proof, we have that%
\begin{equation}
\widetilde{g^{\prime}}(\rho)(t)=g^{\prime}(\rho(t_{0}))\chi_{\{t_{0}%
\}}(t)\oplus...\oplus g^{\prime}(\rho(t_{k}))\chi_{\{t_{k}\}}(t)\oplus...
\end{equation}
is an element of $\overline{P}_{n}$ and%
\[
g(\Phi^{\rho}(\mu,t))=g(\mu\chi_{(-\infty,t_{0})}(t)\oplus\widehat{\Phi
}^{\alpha}(\mu,0)\chi_{\lbrack t_{0},t_{1})}(t)\oplus...
\]%
\[
...\oplus\widehat{\Phi}^{\alpha}(\mu,k)\chi_{\lbrack t_{k},t_{k+1})}%
(t)\oplus...)=
\]%
\[
=g(\mu)\chi_{(-\infty,t_{0})}(t)\oplus g(\widehat{\Phi}^{\alpha}(\mu
,0))\chi_{\lbrack t_{0},t_{1})}(t)\oplus...
\]%
\[
...\oplus g(\widehat{\Phi}^{\alpha}(\mu,k))\chi_{\lbrack t_{k},t_{k+1}%
)}(t)\oplus...=
\]%
\[
=g(\mu)\chi_{(-\infty,t_{0})}(t)\oplus\widehat{\Psi}^{\widehat{g^{\prime}%
}(\alpha)}(g(\mu),0)\chi_{\lbrack t_{0},t_{1})}(t)\oplus...
\]%
\[
...\oplus\widehat{\Psi}^{\widehat{g^{\prime}}(\alpha)}(g(\mu),k)\chi_{\lbrack
t_{k},t_{k+1})}(t)\oplus...=\Psi^{\widetilde{g^{\prime}}(\rho)}(g(\mu),t).
\]

c)$\Longrightarrow$a) Let $\nu,\mu\in\mathbf{B}^{n}$ be arbitrary and fixed
and we consider $\rho\in\overline{P}_{n}$ given by (\ref{sym16}), with
$(t_{k})\in Seq$ fixed, $\rho(t_{0})=\nu$ and $\forall k\geq1,\rho(t_{k}%
)\neq(0,...,0)$. We have%
\begin{equation}
g(\Phi^{\rho}(\mu,t))= \label{ite5}%
\end{equation}%
\[
=g(\mu\chi_{(-\infty,t_{0})}(t)\oplus\Phi^{\nu}(\mu)\chi_{\lbrack t_{0}%
,t_{1})}(t)\oplus\widehat{\Phi}^{\alpha}(\mu,1)\chi_{\lbrack t_{1},t_{2}%
)}(t)\oplus...)=
\]%
\[
=g(\mu)\chi_{(-\infty,t_{0})}(t)\oplus g(\Phi^{\nu}(\mu))\chi_{\lbrack
t_{0},t_{1})}(t)\oplus g(\widehat{\Phi}^{\alpha}(\mu,1))\chi_{\lbrack
t_{1},t_{2})}(t)\oplus...
\]

Case i) $\nu=(0,...,0),$

the commutativity of the diagram is equivalent with the first statement of c).

Case ii) $\nu\neq(0,...,0),$%
\[
\widetilde{g^{\prime}}(\rho)(t)=g^{\prime}(\rho(t))=g^{\prime}(\nu
)\chi_{\{t_{0}\}}(t)\oplus g^{\prime}(\rho(t_{1}))\chi_{\{t_{1}\}}%
(t)\oplus...,
\]%
\[
\Psi^{\widetilde{g^{\prime}}(\rho)}(g(\mu),t)=
\]%
\[
=g(\mu)\chi_{(-\infty,t_{0})}(t)\oplus\Psi^{g^{\prime}(\nu)}(g(\mu
))\chi_{\lbrack t_{0},t_{1})}(t)\oplus\widehat{\Psi}^{\widehat{g^{\prime}%
}(\alpha)}(g(\mu),1)\chi_{\lbrack t_{1},t_{2})}(t)\oplus...
\]
and from (\ref{ite5}), for $t\in\lbrack t_{0},t_{1}),$ we obtain%
\[
g(\Phi^{\nu}(\mu))=\Psi^{g^{\prime}(\nu)}(g(\mu)).
\]

\end{proof}

\begin{definition}
\label{Def64}We consider the functions $\Phi,\Psi:\mathbf{B}^{n}%
\rightarrow\mathbf{B}^{n}.$ If $g,g^{\prime}:\mathbf{B}^{n}\rightarrow
\mathbf{B}^{n}$ bijective exist such that one of the equivalent properties a),
b), c) from Theorem \ref{The29} is satisfied, then we say that the couple
$(g,g^{\prime})$ defines an \textbf{isomorphism} from $\widehat{\overline{\Xi
}}_{\Phi}$ to $\widehat{\overline{\Xi}}_{\Psi},$ or from $\overline{\Xi}%
_{\Phi}$ to $\overline{\Xi}_{\Psi},$ or from $\Phi$ to $\Psi.$ We use the
notation $\overline{Iso}(\Phi,\Psi)$ for the set of these couples and we also
denote with $\overline{Aut}(\Phi)=\overline{Iso}(\Phi,\Phi)$ the set of the
\textbf{automorphisms} of $\widehat{\overline{\Xi}}_{\Phi},$ $\overline{\Xi
}_{\Phi},$ or $\Phi.$
\end{definition}

\begin{theorem}
\label{The28}For $\Phi,\Psi,g,g^{\prime}:\mathbf{B}^{n}\rightarrow
\mathbf{B}^{n}$, the following statements are equivalent:

a) $\forall\nu\in\mathbf{B}^{n},$ the diagram%
\[%
\begin{array}
[c]{ccc}%
\mathbf{B}^{n} & \overset{\Phi^{\nu}}{\rightarrow} & \mathbf{B}^{n}\\
g\downarrow\; &  & \;\downarrow g\\
\mathbf{B}^{n} & \overset{\Psi^{g^{\prime}(\nu)}}{\longleftarrow} &
\mathbf{B}^{n}%
\end{array}
\]
is commutative;

b) $\forall\mu\in\mathbf{B}^{n},\forall\alpha\in\overline{\Pi}_{n},\forall
k\in\mathbf{N}_{\_},$%
\[
g(\mu)=\text{ }^{\ast}\widehat{\Psi}^{\widehat{g^{\prime}}(\alpha)}%
(g(\widehat{\Phi}^{\alpha}(\mu,k)),k);
\]

c) $\forall\mu\in\mathbf{B}^{n},$%
\[
g(\mu)=\Psi^{g^{\prime}(0,...,0)}(g(\mu))
\]
and $\forall\mu\in\mathbf{B}^{n},\forall\rho\in\overline{P}_{n},\forall
t\in\mathbf{R},$%
\[
g(\mu)=\text{ }^{\ast}\Psi^{\widetilde{g^{\prime}}(\rho)}(g(\Phi^{\rho}%
(\mu,t)),t).
\]

\end{theorem}

\begin{proof}
a)$\Longrightarrow$b) We fix arbitrarily $\mu\in\mathbf{B}^{n},$ $\alpha
\in\overline{\Pi}_{n}$ and we use the induction on $k\geq-1.$ In the case
$k=-1$ the equality to be proved is satisfied%
\[
g(\mu)=g(\widehat{\Phi}^{\alpha}(\mu,-1))=\widehat{\Psi}^{\widehat{g^{\prime}%
}(\alpha)}(g(\widehat{\Phi}^{\alpha}(\mu,-1)),-1),
\]
thus we presume that the statement is true for $k$ and we prove it for $k+1.$
We have:%
\[
g(\mu)=\text{ }^{\ast}\widehat{\Psi}^{\widehat{g^{\prime}}(\alpha)}%
(g(\widehat{\Phi}^{\alpha}(\mu,k)),k)=\text{ }^{\ast}\widehat{\Psi}%
^{\widehat{g^{\prime}}(\alpha)}(\Psi^{g^{\prime}(\alpha^{k+1})}(g(\Phi
^{\alpha^{k+1}}(\widehat{\Phi}^{\alpha}(\mu,k)))),k)
\]%
\[
=\text{ }^{\ast}\widehat{\Psi}^{\widehat{g^{\prime}}(\alpha)}(g(\widehat{\Phi
}^{\alpha}(\mu,k+1)),k+1).
\]

The proof is similar with the proof of Theorem \ref{The29}.
\end{proof}

\begin{definition}
\label{Def65}Let be the functions $\Phi,\Psi:\mathbf{B}^{n}\rightarrow
\mathbf{B}^{n}.$ If $g,g^{\prime}:\mathbf{B}^{n}\rightarrow\mathbf{B}^{n}$
bijective exist such that one of the equivalent properties a), b), c) from
Theorem \ref{The28} is fulfilled, we say that the couple $(g,g^{\prime})$
defines an \textbf{anti-isomorphism} from $\widehat{\overline{\Xi}}_{\Phi}$ to
$^{\ast}\widehat{\overline{\Xi}}_{\Psi},$ or from $\overline{\Xi}_{\Phi}$ to
$^{\ast}\overline{\Xi}_{\Psi},$ or from $\Phi$ to $\Psi.$ We use the notation
$^{\ast}\overline{Iso}(\Phi,\Psi)$ for these couples and we also denote with
$^{\ast}\overline{Aut}(\Phi)=$ $^{\ast}\overline{Iso}(\Phi,\Phi)$ the set of
the \textbf{anti-automorphisms} of $\widehat{\overline{\Xi}}_{\Phi},$
$\overline{\Xi}_{\Phi}$ or $\Phi$.
\end{definition}

\section{Symmetry and anti-symmetry}

\begin{remark}
The fact that $(1_{\mathbf{B}^{n}},1_{\mathbf{B}^{n}})\in\overline{Aut}(\Phi)$
implies $\overline{Aut}(\Phi)\neq\emptyset,$ but all of $\overline{Iso}%
(\Phi,\Psi),^{\ast}\overline{Iso}(\Phi,\Psi)$ and $^{\ast}\overline{Aut}%
(\Phi)$ may be empty.
\end{remark}

\begin{definition}
\label{Def66}Let be $\Phi,\Psi:\mathbf{B}^{n}\rightarrow\mathbf{B}^{n},$
$\Phi\neq\Psi.$ If $\overline{Iso}(\Phi,\Psi)\neq\emptyset$, then
$\widehat{\overline{\Xi}}_{\Phi},$ $\widehat{\overline{\Xi}}_{\Psi};$
$\overline{\Xi}_{\Phi},\overline{\Xi}_{\Psi};$ $\Phi,\Psi$ are called
\textbf{symmetrical}, or \textbf{conjugated}; if $^{\ast}\overline{Iso}%
(\Phi,\Psi)\neq\emptyset$, then $\widehat{\overline{\Xi}}_{\Phi},$ $^{\ast
}\widehat{\overline{\Xi}}_{\Psi};$ $\overline{\Xi}_{\Phi},^{\ast}\overline
{\Xi}_{\Psi};$ $\Phi,\Psi$ are called \textbf{anti-symmetrical}, or
\textbf{anti-conjugated}.

If $card(\overline{Aut}(\Phi))>1,$ then $\widehat{\overline{\Xi}}_{\Phi},$
$\overline{\Xi}_{\Phi}$ and $\Phi$ are called \textbf{symmetrical} and if
$^{\ast}\overline{Aut}(\Phi)\neq\emptyset,$ then $\widehat{\overline{\Xi}%
}_{\Phi},$ $\overline{\Xi}_{\Phi}$ and $\Phi$ are called
\textbf{anti-symmetrical}.
\end{definition}

\begin{remark}
\label{Rem33}The symmetry of $\Phi,\Psi$ means that $(g,g^{\prime}%
)\in\overline{Iso}(\Phi,\Psi)$ maps the transfers $\mu\rightarrow\Phi^{\nu
}(\mu)$ in transfers $g(\mu)\rightarrow g(\Phi^{\nu}(\mu))=\Psi^{g^{\prime
}(\nu)}(g(\mu));$ the situation when $\Phi$ is symmetrical and $(g,g^{\prime
})\in\overline{Aut}(\Phi)$ is similar. Anti-symmetry may be understood as
mirroring: $(g,g^{\prime})\in$ $^{\ast}\overline{Iso}(\Phi,\Psi)$ maps the
transfers (or arrows) $\mu\rightarrow\Phi^{\nu}(\mu)$ in transfers
$g(\mu)\longleftarrow g(\Phi^{\nu}(\mu))=\Psi^{g^{\prime}(\nu)}(g(\mu))$ and
similarly for $(g,g^{\prime})\in$ $^{\ast}\overline{Aut}(\Phi).$
\end{remark}

\begin{theorem}
\label{The32}Let be $\Phi,\Psi:\mathbf{B}^{n}\rightarrow\mathbf{B}^{n}.$

a) If $(g,g^{\prime})\in\overline{Iso}(\Phi,\Psi)$, then $(g^{-1},g^{\prime
-1})\in\overline{Iso}(\Psi,\Phi).$

b) If $(g,g^{\prime})\in$ $^{\ast}\overline{Iso}(\Phi,\Psi)$, then
$(g^{-1},g^{\prime-1})\in$ $^{\ast}\overline{Iso}(\Psi,\Phi).$
\end{theorem}

\begin{proof}
a) The hypothesis states that $\forall\nu\in\mathbf{B}^{n},$ the diagram%
\[%
\begin{array}
[c]{ccc}%
\mathbf{B}^{n} & \overset{\Phi^{\nu}}{\rightarrow} & \mathbf{B}^{n}\\
g\downarrow\; &  & \;\downarrow g\\
\mathbf{B}^{n} & \overset{\Psi^{g^{\prime}(\nu)}}{\rightarrow} &
\mathbf{B}^{n}%
\end{array}
\]
commutes, with $g,g^{\prime}$ bijective$.$ We fix arbitrarily $\nu
\in\mathbf{B}^{n},\mu\in\mathbf{B}^{n}.$ We denote $\mu^{\prime}=g(\mu
),\nu^{\prime}=g^{\prime}(\nu)$ and we note that%
\begin{equation}
g^{-1}(\Psi^{\nu^{\prime}}(\mu^{\prime}))=\Phi^{g^{\prime-1}(\nu^{\prime}%
)}(g^{-1}(\mu^{\prime})). \label{ite11}%
\end{equation}
As $\nu,\mu$ were chosen arbitrarily and on the other hand, when $\nu$ runs in
$\mathbf{B}^{n},\nu^{\prime}$ runs in $\mathbf{B}^{n}$ and when $\mu$ runs in
$\mathbf{B}^{n},\mu^{\prime}$ runs in $\mathbf{B}^{n}$, we infer that
(\ref{ite11}) is equivalent with the commutativity of the diagram%
\[%
\begin{array}
[c]{ccc}%
\;\;\;\;\;\mathbf{B}^{n} & \overset{\Psi^{\nu^{\prime}}}{\rightarrow} &
\mathbf{B}^{n}\;\;\;\\
g^{-1}\downarrow &  & \downarrow g^{-1}\\
\;\;\;\;\;\mathbf{B}^{n} & \overset{\Phi^{g^{\prime-1}(\nu^{\prime})}%
}{\rightarrow} & \mathbf{B}^{n}\;\;\;
\end{array}
\]
for any $\nu^{\prime}\in\mathbf{B}^{n}.$ We have proved that $(g^{-1}%
,g^{\prime-1})\in\overline{Iso}(\Psi,\Phi).$

b) By hypothesis $\forall\nu\in\mathbf{B}^{n},$ the diagram%
\[%
\begin{array}
[c]{ccc}%
\mathbf{B}^{n} & \overset{\Phi^{\nu}}{\rightarrow} & \mathbf{B}^{n}\\
g\downarrow\; &  & \;\downarrow g\\
\mathbf{B}^{n} & \overset{\Psi^{g^{\prime}(\nu)}}{\longleftarrow} &
\mathbf{B}^{n}%
\end{array}
\]
is commutative, $g,g^{\prime}$ bijective and we prove that $\forall\nu
^{\prime}\in\mathbf{B}^{n},$ the diagram%
\[%
\begin{array}
[c]{ccc}%
\;\;\;\;\;\mathbf{B}^{n} & \overset{\Psi^{\nu^{\prime}}}{\rightarrow} &
\mathbf{B}^{n}\;\;\;\\
g^{-1}\downarrow &  & \downarrow g^{-1}\\
\;\;\;\;\;\mathbf{B}^{n} & \overset{\Phi^{g^{\prime-1}(\nu^{\prime})}%
}{\longleftarrow} & \mathbf{B}^{n}\;\;\;
\end{array}
\]
is commutative.
\end{proof}

\begin{theorem}
\label{The36}$\overline{Aut}(\Phi)$ is a group relative to the law:
$\forall(g,g^{\prime})\in\overline{Aut}(\Phi),$ $\forall(h,h^{\prime}%
)\in\overline{Aut}(\Phi),$%
\[
(h,h^{\prime})\circ(g,g^{\prime})=(h\circ g,h^{\prime}\circ g^{\prime}).
\]

\end{theorem}

\begin{proof}
The fact that $\forall(g,g^{\prime})\in\overline{Aut}(\Phi),\forall
(h,h^{\prime})\in\overline{Aut}(\Phi),(h\circ g,h^{\prime}\circ g^{\prime}%
)\in\overline{Aut}(\Phi)$ is proved like this: $\forall\nu\in\mathbf{B}^{n},$%
\[
(h\circ g)\circ\Phi^{\nu}=h\circ(g\circ\Phi^{\nu})=h\circ(\Phi^{g^{\prime}%
(\nu)}\circ g)=(h\circ\Phi^{g^{\prime}(\nu)})\circ g=
\]%
\[
=(\Phi^{h^{\prime}(g^{\prime}(\nu))}\circ h)\circ g=\Phi^{(h^{\prime}\circ
g^{\prime})(\nu)}\circ(h\circ g);
\]
the fact that $(1_{\mathbf{B}^{n}},1_{\mathbf{B}^{n}})\in\overline{Aut}(\Phi)$
was mentioned before; and the fact that $\forall(g,g^{\prime})\in
\overline{Aut}(\Phi),$ $(g^{-1},g^{\prime-1})\in\overline{Aut}(\Phi)$ was
shown at Theorem \ref{The32} a).
\end{proof}

\begin{definition}
Any subgroup $G\subset\overline{Aut}(\Phi)$ with $card(G)>1$\ is called a
\textbf{group of symmetry} of $\widehat{\overline{\Xi}}_{\Phi},$ of
$\overline{\Xi}_{\Phi}$ or of $\Phi.$
\end{definition}

\section{Examples}

\begin{example}
\label{Exa6_}$\Phi,\Psi:\mathbf{B}^{2}\rightarrow\mathbf{B}^{2}$ are given by,
see Figure \ref{echiv1}
\begin{figure}
[ptb]
\begin{center}
\fbox{\includegraphics[
height=1.4313in,
width=3.1981in
]%
{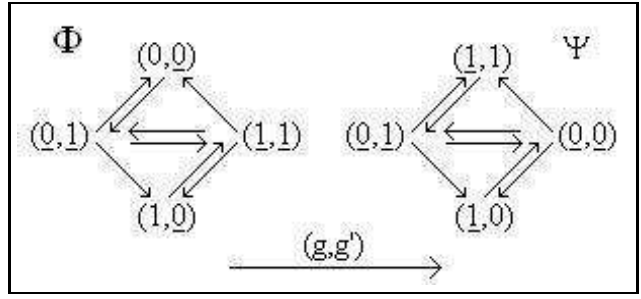}%
}\caption{Symmetrical systems, Example \ref{Exa6_}}%
\label{echiv1}%
\end{center}
\end{figure}
\[
\forall(\mu_{1},\mu_{2})\in\mathbf{B}^{2},\Phi(\mu_{1},\mu_{2})=(\mu_{1}%
\oplus\mu_{2},\overline{\mu_{2}}),
\]%
\[
\forall(\mu_{1},\mu_{2})\in\mathbf{B}^{2},\Psi(\mu_{1},\mu_{2})=(\overline
{\mu_{1}},\overline{\mu_{1}}\;\overline{\mu_{2}}\cup\mu_{1}\mu_{2})
\]
and the bijections $g,g^{\prime}:\mathbf{B}^{2}\rightarrow\mathbf{B}^{2}$ are
$\forall(\mu_{1},\mu_{2})\in\mathbf{B}^{2},$%
\[
g(\mu_{1},\mu_{2})=(\overline{\mu_{2}},\overline{\mu_{1}}),
\]%
\[
g^{\prime}(\mu_{1},\mu_{2})=(\mu_{2},\mu_{1})
\]
(in order to understand the choice of $g^{\prime},$ to be remarked in Figure
\ref{echiv1} the positions of the underlined coordinates for $\Phi$ and $\Psi
$). $\Phi$ and $\Psi$ are conjugated.
\end{example}

\begin{example}
\label{Exe2}The system from Figure \ref{simetrie15} is symmetrical and a group
of symmetry is generated by the couples $(g,1_{\mathbf{B}^{3}}%
),(u,1_{\mathbf{B}^{3}}),(v,1_{\mathbf{B}^{3}}),$ see Table 2; $g,u,v$ are
transpositions that permute the isolated fixed points
$(1,0,0),(1,0,1),(1,1,1).$
\begin{figure}
[ptb]
\begin{center}
\fbox{\includegraphics[
height=0.8441in,
width=3.6218in
]%
{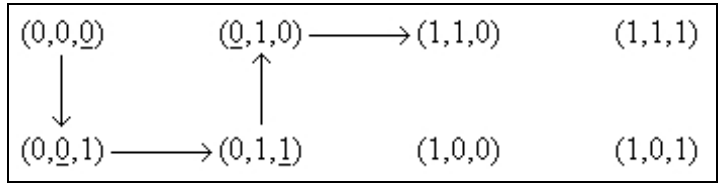}%
}\caption{Symmetrical system, Example \ref{Exe2}}%
\label{simetrie15}%
\end{center}
\end{figure}
\begin{align*}
&
\begin{array}
[c]{ccccc}%
(\mu_{1},\mu_{2},\mu_{3}) & 1_{\mathbf{B}^{3}} & g & u & v\\
(0,0,0) & (0,0,0) & (0,0,0) & (0,0,0) & (0,0,0)\\
(0,0,1) & (0,0,1) & (0,0,1) & (0,0,1) & (0,0,1)\\
(0,1,0) & (0,1,0) & (0,1,0) & (0,1,0) & (0,1,0)\\
(0,1,1) & (0,1,1) & (0,1,1) & (0,1,1) & (0,1,1)\\
(1,0,0) & (1,0,0) & (1,0,0) & (1,0,1) & (1,1,1)\\
(1,0,1) & (1,0,1) & (1,1,1) & (1,0,0) & (1,0,1)\\
(1,1,0) & (1,1,0) & (1,1,0) & (1,1,0) & (1,1,0)\\
(1,1,1) & (1,1,1) & (1,0,1) & (1,1,1) & (1,0,0)
\end{array}
\\
&  \;\;\;\;\;\;\;\;\;\;\;\;\;\;\;\;\;\;\;\;\;\;\;\;\;Table\;2
\end{align*}

\end{example}

\begin{example}
\label{Exa11}The function $\Phi:\mathbf{B}^{2}\rightarrow\mathbf{B}^{2}$
defined by $\forall\mu\in\mathbf{B}^{2},\Phi(\mu_{1},\mu_{2})=(\overline
{\mu_{1}},\overline{\mu_{2}})$ fulfills for $\nu\in\mathbf{B}^{2}:$%
\[
\Phi^{\nu}(\mu_{1},\mu_{2})=(\overline{\nu_{1}}\mu_{1}\oplus\nu_{1}%
\overline{\mu_{1}},\overline{\nu_{2}}\mu_{2}\oplus\nu_{2}\overline{\mu_{2}}),
\]%
\[
(\Phi^{\nu}\circ\Phi^{\nu})(\mu_{1},\mu_{2})=(\overline{\nu_{1}}\Phi_{1}%
^{\nu_{1}}(\mu_{1},\mu_{2})\oplus\nu_{1}\overline{\Phi_{1}^{\nu_{1}}(\mu
_{1},\mu_{2})}),
\]%
\[
\overline{\nu_{2}}\Phi_{2}^{\nu_{2}}(\mu_{1},\mu_{2})\oplus\nu_{2}%
\overline{\Phi_{2}^{\nu_{2}}(\mu_{1},\mu_{2})})
\]%
\[
=(\overline{\nu_{1}}(\overline{\nu_{1}}\mu_{1}\oplus\nu_{1}\overline{\mu_{1}%
})\oplus\nu_{1}(\overline{\nu_{1}}\mu_{1}\oplus\nu_{1}\overline{\mu_{1}}%
\oplus1),
\]%
\[
\overline{\nu_{2}}(\overline{\nu_{2}}\mu_{2}\oplus\nu_{2}\overline{\mu_{2}%
})\oplus\nu_{2}(\overline{\nu_{2}}\mu_{2}\oplus\nu_{2}\overline{\mu_{2}}%
\oplus1))
\]%
\[
=((\nu_{1}\oplus1)\mu_{1}\oplus\nu_{1}(\mu_{1}\oplus1)\oplus\nu_{1},(\nu
_{2}\oplus1)\mu_{2}\oplus\nu_{2}(\mu_{2}\oplus1)\oplus\nu_{2})
\]%
\[
=(\nu_{1}\mu_{1}\oplus\mu_{1}\oplus\nu_{1}\mu_{1}\oplus\nu_{1}\oplus\nu
_{1},\nu_{2}\mu_{2}\oplus\mu_{2}\oplus\nu_{2}\mu_{2}\oplus\nu_{2}\oplus\nu
_{2})=(\mu_{1},\mu_{2}),
\]
thus $(1_{\mathbf{B}^{2}},1_{\mathbf{B}^{2}})\in$ $^{\ast}\overline{Aut}%
(\Phi)$ and $\Phi$ is anti-symmetrical. The state portrait of $\Phi$ was drawn
in Figure \ref{simetrie30} c).
\end{example}

\begin{notation}
Let $\sigma:\{1,...,n\}\rightarrow\{1,...,n\}$ be a bijection. We use the
notation $\pi_{\sigma}:\mathbf{B}^{n}\rightarrow\mathbf{B}^{n}$ for the
bijection given by $\forall\mu\in\mathbf{B}^{n},$%
\[
\pi_{\sigma}(\mu_{1},...,\mu_{n})=(\mu_{\sigma(1)},...,\mu_{\sigma(n)}).
\]

\end{notation}

\begin{definition}
Any of $\widehat{\overline{\Xi}}_{\Phi},$ $\overline{\Xi}_{\Phi}$ and
$\Phi:\mathbf{B}^{n}\rightarrow\mathbf{B}^{n}$ is called \textbf{symmetrical
relative to the coordinates }if the bijection $\sigma$ exists, $\sigma
\neq1_{\{1,...,n\}}$ such that $(\pi_{\sigma},\pi_{\sigma})\in\overline
{Aut}(\Phi).$
\end{definition}

\begin{example}
\label{Exa17}We consider the function $\Phi:\mathbf{B}^{3}\rightarrow
\mathbf{B}^{3}$ defined by $\forall\mu\in\mathbf{B}^{3},$ $\Phi(\mu_{1}%
,\mu_{2},\mu_{3})=(\mu_{2}\mu_{3}\oplus\mu_{1}\oplus\mu_{2},\mu_{1}\mu
_{3}\oplus\mu_{2}\oplus\mu_{3},\mu_{1}\mu_{2}\oplus\mu_{1}\oplus\mu_{3})$ and
the permutation $\sigma:\{1,2,3\}\rightarrow\{1,2,3\},$ $\sigma=\left(
\begin{array}
[c]{ccc}%
1 & 2 & 3\\
\sigma(1) & \sigma(2) & \sigma(3)
\end{array}
\right)  =\left(
\begin{array}
[c]{ccc}%
1 & 2 & 3\\
3 & 1 & 2
\end{array}
\right)  .$ A group of symmetry of $\overline{\Xi}_{\Phi}$ is represented by
$G=\{(1_{\mathbf{B}^{3}},1_{\mathbf{B}^{3}}),$ $(\pi_{\sigma},\pi_{\sigma
}),(\pi_{\sigma\circ\sigma},\pi_{\sigma\circ\sigma})\}.$ We have given in
Figure \ref{simetrie14} the state portrait of $\Phi.$
\begin{figure}
[ptb]
\begin{center}
\fbox{\includegraphics[
height=0.8683in,
width=3.8424in
]%
{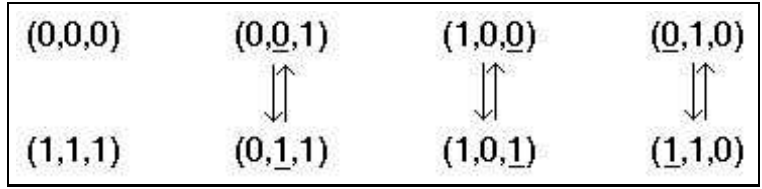}%
}\caption{System that is symmetrical relative to the coordinates, Example
\ref{Exa17}}%
\label{simetrie14}%
\end{center}
\end{figure}

\end{example}

\begin{notation}
For $\lambda\in\mathbf{B}^{n},$ we denote by $\theta^{\lambda}:\mathbf{B}%
^{n}\rightarrow\mathbf{B}^{n}$ the translation of vector $\lambda:$
$\forall\mu\in\mathbf{B}^{n},$%
\[
\theta^{\lambda}(\mu)=\mu\oplus\lambda.
\]

\end{notation}

\begin{definition}
If $(\theta^{\lambda},g^{\prime})\in\overline{Aut}(\Phi)$ holds for some
$(\theta^{\lambda},g^{\prime})\neq(1_{\mathbf{B}^{n}},1_{\mathbf{B}^{n}})$, we
say that any of $\widehat{\overline{\Xi}}_{\Phi},$ $\overline{\Xi}_{\Phi}$ and
$\Phi$ is \textbf{symmetrical relative to translations}.
\end{definition}

\begin{example}
\label{Exa13}In Figure \ref{simetrie10}
\begin{figure}
[ptb]
\begin{center}
\fbox{\includegraphics[
height=1.4702in,
width=2.6074in
]%
{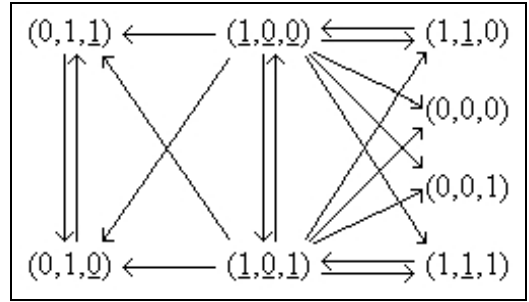}%
}\caption{$\Phi$ has the automorphism $(\theta^{(0,0,1)},1_{\mathbf{B}^{3}}),$
Example \ref{Exa13}}%
\label{simetrie10}%
\end{center}
\end{figure}
we have the system with $\Phi$ given by Table 3%
\begin{align*}
&
\begin{array}
[c]{cc}%
(\mu_{1},\mu_{2},\mu_{3}) & \Phi\\
(0,0,0) & (0,0,0)\\
(0,0,1) & (0,0,1)\\
(0,1,0) & (0,1,1)\\
(0,1,1) & (0,1,0)\\
(1,0,0) & (0,1,1)\\
(1,0,1) & (0,1,0)\\
(1,1,0) & (1,0,0)\\
(1,1,1) & (1,0,1)
\end{array}
\\
&  \;\;\;\;\;\;\;\;\;Table\;3
\end{align*}
and $(\theta^{(0,0,1)},1_{\mathbf{B}^{3}})\in\overline{Aut}(\Phi),$ as
resulting from the state portrait.
\end{example}

\begin{example}
\label{Exa14}In Table 4 we have a function $\Phi:\mathbf{B}^{2}\rightarrow
\mathbf{B}^{2}$ for which four functions $g_{1}^{\prime},g_{2}^{\prime}%
,g_{3}^{\prime},g_{4}^{\prime}:\mathbf{B}^{2}\rightarrow\mathbf{B}^{2}$ exist:%
\begin{align*}
&
\begin{array}
[c]{cccccc}%
(\mu_{1},\mu_{2}) & \Phi & g_{1}^{\prime} & g_{2}^{\prime} & g_{3}^{\prime} &
g_{4}^{\prime}\\
(0,0) & (0,0) & (0,0) & (1,0) & (0,0) & (1,0)\\
(0,1) & (0,1) & (0,1) & (0,1) & (1,1) & (1,1)\\
(1,0) & (1,1) & (1,0) & (0,0) & (1,0) & (0,0)\\
(1,1) & (1,0) & (1,1) & (1,1) & (0,1) & (0,1)
\end{array}
\\
&  \;\;\;\;\;\;\;\;\;\;\;\;\;\;\;\;\;\;\;Table\;4
\end{align*}
such that $(1_{\mathbf{B}^{2}},g_{1}^{\prime}),(1_{\mathbf{B}^{2}}%
,g_{2}^{\prime}),(1_{\mathbf{B}^{2}},g_{3}^{\prime}),(1_{\mathbf{B}^{2}}%
,g_{4}^{\prime})\in\overline{Aut}(\Phi).$ The state portrait of $\Phi$ is
drawn Figure \ref{simetrie8}.%
\begin{figure}
[ptb]
\begin{center}
\fbox{\includegraphics[
height=0.6382in,
width=1.2929in
]%
{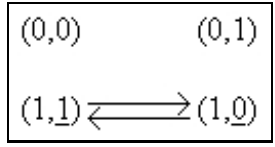}%
}\caption{$\Phi$ is symmetrical relative to translations with $(0,0),$ Example
\ref{Exa14}}%
\label{simetrie8}%
\end{center}
\end{figure}

\end{example}

\begin{example}
\label{Exe1}The system from Figure \ref{simetrie1} is symmetrical relative to
translations,%
\begin{figure}
[ptb]
\begin{center}
\fbox{\includegraphics[
height=0.6988in,
width=1.292in
]%
{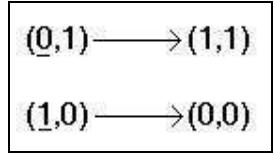}%
}\caption{Function $\Phi$ that is self dual, $(\theta^{(1,1)},1_{\mathbf{B}%
^{2}})\in\overline{Aut}(\Phi),$ Example \ref{Exe1}}%
\label{simetrie1}%
\end{center}
\end{figure}
since it has the group of symmetry $G=\{(1_{\mathbf{B}^{2}},1_{\mathbf{B}^{2}%
}),(\theta^{(1,1)},1_{\mathbf{B}^{2}})\}.$ $\Phi$ is self-dual $\Phi
=\Phi^{\ast},$ where the dual $\Phi^{\ast}$ of $\Phi$ is defined by
$\Phi^{\ast}(\mu)=\overline{\Phi(\overline{\mu})}.$
\end{example}

\begin{example}
\label{Exa15}Functions $\Phi:\mathbf{B}^{2}\rightarrow\mathbf{B}^{2}$ exist,
see Figure \ref{simetrie7}
\begin{figure}
[ptb]
\begin{center}
\fbox{\includegraphics[
height=2.4621in,
width=3.0191in
]%
{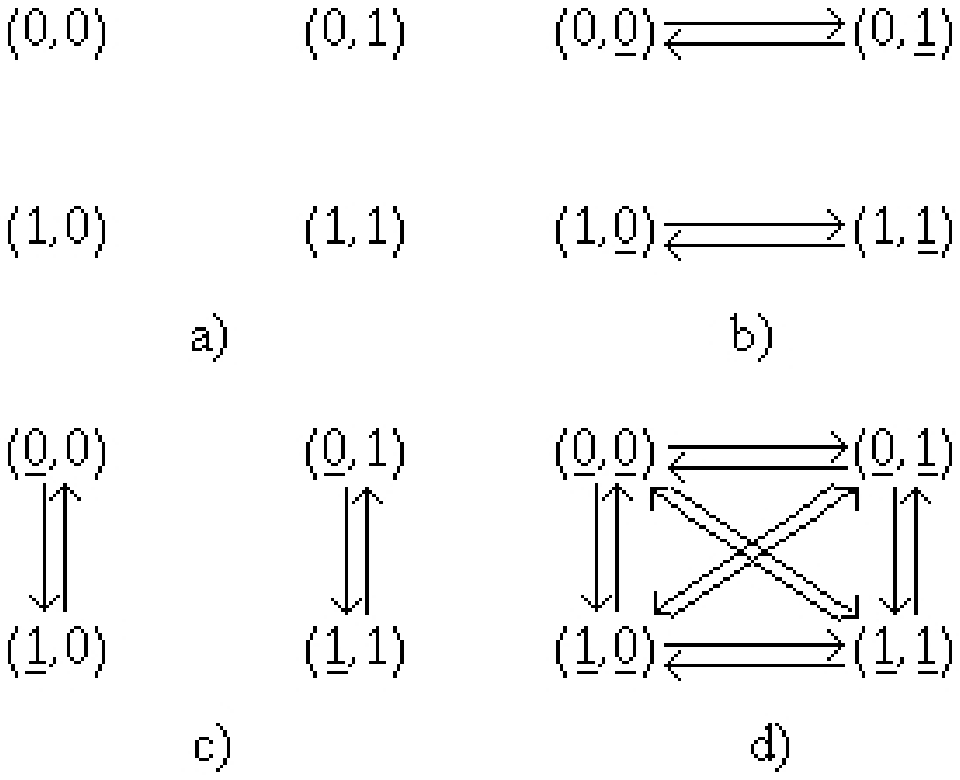}%
}\caption{Functions $\Phi$ that are self dual, $(\theta^{(1,1)},1_{\mathbf{B}%
^{2}})\in\overline{Aut}(\Phi),$ Example \ref{Exa15}}%
\label{simetrie7}%
\end{center}
\end{figure}
that are invariant relative to the translations with any $\lambda\in
\mathbf{B}^{2},$ thus their group of symmetry is $G=\{(1_{\mathbf{B}^{2}%
},1_{\mathbf{B}^{2}}),(\theta^{(0,1)},1_{\mathbf{B}^{2}}),(\theta
^{(1,0)},1_{\mathbf{B}^{2}}),(\theta^{(1,1)},1_{\mathbf{B}^{2}})\}.$ The fact
that $(\theta^{(1,1)},1_{\mathbf{B}^{2}})\in G$ shows that all these
functions: $\Phi(\mu)=(\mu_{1},\mu_{2}),$ $\Phi(\mu)=(\mu_{1},\overline
{\mu_{2}}),$ $\Phi(\mu)=(\overline{\mu_{1}},\mu_{2}),$ $\Phi(\mu
)=(\overline{\mu_{1}},\overline{\mu_{2}})$ are self-dual, $\Phi=\Phi^{\ast}.$
\end{example}

\begin{example}
\label{Exe4}The group of symmetry $G$ of the system from Figure
\ref{simetrie4}
\begin{figure}
[ptb]
\begin{center}
\fbox{\includegraphics[
height=1.0672in,
width=1.8144in
]%
{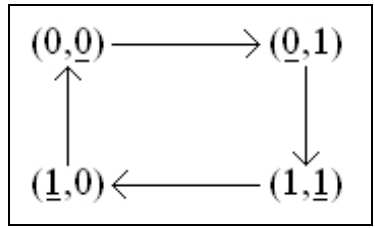}%
}\caption{Symmetry including symmetry relative to translations, Example
\ref{Exe4}}%
\label{simetrie4}%
\end{center}
\end{figure}
has four elements given by%
\begin{align*}
&
\begin{array}
[c]{ccccc}%
(\mu_{1},\mu_{2}) & 1_{\mathbf{B}^{2}} & g & h & \theta^{(1,1)}\\
(0,0) & (0,0) & (0,1) & (1,0) & (1,1)\\
(0,1) & (0,1) & (1,1) & (0,0) & (1,0)\\
(1,0) & (1,0) & (0,0) & (1,1) & (0,1)\\
(1,1) & (1,1) & (1,0) & (0,1) & (0,0)
\end{array}
\\
&  \;\;\;\;\;\;\;\;\;\;\;\;\;\;\;\;Table\;5
\end{align*}
and we remark that $h=g^{-1},\theta^{(1,1)}=(\theta^{(1,1)})^{-1}$ hold. On
the other hand%
\begin{align*}
&
\begin{array}
[c]{ccccc}%
(\nu_{1},\nu_{2}) & (_{1_{\mathbf{B}^{2}}})^{\prime} & g^{\prime} & h^{\prime}
& (\theta^{(1,1)})^{\prime}\\
(0,0) & (0,0) & (0,0) & (0,0) & (0,0)\\
(0,1) & (0,1) & (1,0) & (1,0) & (0,1)\\
(1,0) & (1,0) & (0,1) & (0,1) & (1,0)\\
(1,1) & (1,1) & (1,1) & (1,1) & (1,1)
\end{array}
\\
&  \;\;\;\;\;\;\;\;\;\;\;\;\;\;\;\;Table\;6
\end{align*}
$G$ has a proper subgroup $G^{\prime}=\{(1_{\mathbf{B}^{2}},1_{\mathbf{B}^{2}%
}),(\theta^{(1,1)},1_{\mathbf{B}^{2}})\},$ showing that $\Phi=\Phi^{\ast}$
like previously.
\end{example}

\section{Conclusions}

The paper defines the universal semi-regular autonomous asynchronous systems
and the universal anti-semi-regular autonomous asynchronous systems. It also
defines and characterizes the isomorphisms (automorphisms) and the
anti-isormorphisms (anti-automorphisms) of these systems. Symmetry is defined
as the existence of such \ isomorphisms (automorphisms), while anti-symmetry
is defined as the existence of such anti-isomorphisms (anti-automorphisms).
Many examples are given. A by-pass product in this study is anti-symmetry,
that is related with systems having the cause in the future and the effect in
the present. Another by-pass product consists in semi-regularity, since
important examples of isomorphisms (automorphisms) are of semi-regular systems
only, they do not keep progressiveness and regularity \cite{bib2}, \cite{bib3}.

\end{document}